\documentclass[pdflatex,sn-mathphys-num]{sn-jnl}

\usepackage{graphicx}%
\usepackage{multirow}%
\usepackage{amsmath,amssymb,amsfonts}%
\usepackage{amsthm}%
\usepackage{mathrsfs}%
\usepackage[title]{appendix}%
\usepackage{xcolor}%
\usepackage{textcomp}%
\usepackage{manyfoot}%
\usepackage{booktabs}%
\usepackage{algorithm}%
\usepackage{algorithmicx}%
\usepackage{algpseudocode}%
\usepackage{listings}%
\usepackage{amsmath,amscd}
\usepackage{amssymb, stmaryrd}
\usepackage{url}
\usepackage{comment}
\usepackage[whileod]{clr-alg}

\colorlet{shadecolor}{lightgray!30} 
\usepackage{framed}

\usepackage{color}
\usepackage{tikz}
\usetikzlibrary{automata} 
\usetikzlibrary{positioning} 
\usetikzlibrary{arrows} 
\tikzset{node distance=1.5cm, 
            every state/.style={ 
                  semithick,
                  fill=gray!10,
                  minimum size=7mm},
            initial text={}, 
            double distance=2pt, 
            every edge/.style={ 
                   draw,
                   ->,>=stealth', 
                   auto,
            semithick}}

\let\epsilon\varepsilon
\theoremstyle{thmstyleone}%
\newtheorem{theorem}{Theorem}
\newtheorem{proposition}[theorem]{Proposition}%
\newtheorem{lemma}[theorem]{Lemma}%

\theoremstyle{thmstyletwo}%
\newtheorem{example}{Example}%

\theoremstyle{thmstylethree}%

\newcommand{\Suff}{\mbox{Suff\/}}
\newcommand{\DAWG}{\mbox{\it DAWG\/}}
\newcommand{\ttA}{A} 
\newcommand{\tta}{\mathtt{a}} %
\newcommand{\ttb}{\mathtt{b}}
\newcommand{\DONE}{\mathit{checked}}

\DeclareMathOperator{\EP}{endpos} 
\DeclareMathOperator{\pos}{\mathit{POS}}

\usepackage{xspace}

\definecolor{lime}{HTML}{A6CE39}
\DeclareRobustCommand{\orcidicon}{%
	\begin{tikzpicture}
	\draw[lime, fill=lime] (0,0)
	circle [radius=0.16]
	node[white] {{\fontfamily{qag}\selectfont \tiny ID}};
	\draw[white, fill=white] (-0.0625,0.095)
	circle [radius=0.007];
	\end{tikzpicture}
	\hspace{-2mm}
}
\foreach \x in {A, ..., Z}{%
 \expandafter\xdef\csname orcid\x\endcsname{\noexpand%
 \href{https://orcid.org/\csname orcidauthor\x\endcsname}{\noexpand\orcidicon}}
}

\begin{document}
\title[Article Title]{Checking and producing word attractors}

\author*[1]{\sur{Marie-Pierre B\'eal\orcidA{}}}\email{marie-pierre.beal@univ-eiffel.fr}

\author*[1,2]{\sur{Maxime Crochemore\orcidB{}}}\email{maxime.crochemore@univ-eiffel.fr}

\author*[3]{ \sur{Giuseppe Romana\orcidC{}}}\email{giuseppe.romana01@unipa.it}

\affil*[1]{\orgname{Univ. Gustave Eiffel, CNRS, LIGM}, \orgaddress{\street{5 Bd Descartes}, \city{Marne-la-Vall\'ee}, \postcode{F-77454}, \country{France}}}

\affil*[2]{\orgname{King's College London}, \orgaddress{\street{Aldwych}, \city{London},
\postcode{WC2B 4BG}, \country{UK}}
}

\affil*[3]{\orgdiv{Dipartimento di Matematica e Informatica}, \orgname{Univ. Palermo}, \orgaddress{\street{Via Archirafi 34}, \city{Palermo}, \postcode{34 - 90123}, \country{Italy}}}
\abstract{
The article focuses on word (or string) attractors, which are sets of positions related to the text compression efficiency of the underlying word. The article presents two combinatorial algorithms based on Suffix automata or Directed Acyclic Word Graphs. The first algorithm decides in linear time whether a set of positions on the word is an attractor of the word. The second algorithm generates an attractor for a given word in a greedy manner. Although the problem of generating a smallest attractor is NP-hard, {the algorithm is efficient and produces very small attractors and even the smallest ones for several well-known infinite families of words.}
}

\keywords{Word attractor, string attractor, Directed Acyclic Word Graph, Suffix automaton, Suffix links}
\maketitle
\section{Introduction}
The notion of a \emph{word attractor}, or simply an attractor, provides a unifying framework for known dictionary-based compressors, as introduced by Kempa and Prezza in \cite{KempaP18}. 
A word attractor for a finite word \( w \) is a subset \( P \) of positions on \( w \) for which every non-empty factor \( u \) of \( w \) has an occurrence that includes at least one position in \( P \).
That is, there are positions $i$,
$j$ and $t$ that satisfy $u=w_{[i,j]}$, $t\in P$ and $i\leq t\leq j$.

The combinatorial aspects of word attractors have been studied by various authors from a general combinatorics on words perspective and to describe their properties with respect to specific words \cite{MantaciRRRS21,DBLP:conf/latin/RestivoRS22,KutsukakeMNIBT20,
DBLP:conf/cwords/GheeraertRS23,DBLP:journals/jcta/CassaigneGRRSS24}. The notion is also considered for infinite words as a whole, rather just finite prefixes of infinite words, see for example \cite{DBLP:journals/corr/abs-2403-13449}.

Several articles, such as \cite{KempaP18,KempaPPR18, DBLP:conf/ictcs/Romana23}, consider testing whether a set of positions on a word is a word attractor for it and provide efficient solutions.
However, finding a word attractor of minimum size turns out to be NP-hard \cite{KempaP18}, even if the length of factors is bounded by a constant. The same holds for the variant of circular word attractors introduced by Mantaci et al. \cite{MantaciRRRS21,DBLP:conf/ictcs/Romana23}.

Checking whether a set of positions is a word attractor of a word can be performed in linear time in the size of the word \cite{KempaP18,KempaPPR18, DBLP:conf/ictcs/Romana23}. 
Specifically, for a word of length $n$ on an alphabet $A$, the algorithms in \cite{KempaP18, KempaPPR18} run in $O(n)$ time with
$O(n (\log |A| +\log n))$ bit space for integer alphabets of polynomial size in $n$, becoming $O(n \log |A|)$ bit space when $n$ is polynomial in the size of the alphabet.
The algorithm in \cite{DBLP:conf/ictcs/Romana23} runs in $O(n)$ time with $O(n \log n)$ bit space for integer alphabets of polynomial size in $n$.
However, these algorithms use recent advances in compact data structures: compressed suffix arrays supporting constant-time LF (Last-to-First) function computation for the algorithm in \cite{KempaPPR18}, and suffix arrays with the LCP (Longest Common Prefix) array in~\cite{DBLP:conf/ictcs/Romana23}.

In this paper, Section~\ref{sect-checking} revisits the algorithmic question of whether a set of positions on a word is an attractor for the word, using only the Suffix automaton of the word, also called its Directed Acyclic Word Graph (DAWG) (\cite{BlumerBHECS85,CHL07cup}), and its tree of suffix links.
The condition we use to test an attractor is the same as the condition previously considered in [7]. But the data structure is rather different, which makes the whole strategy much clearer. Although suffix arrays can be converted into Suffix automata and vice versa, our algorithm is simpler because it checks the attractor property with a direct and simple computation on the DAWG.

The time complexity for computing the DAWG with its suffix links is $O(n)$ for a fixed-size alphabet, $O(n \log |A|)$ for a general ordered alphabet $A$, and $O(n)$ in \cite{FujishigeTIBT23} for an integer alphabet of polynomial size. 
Assuming that the DAWG and its attributes are computed, our algorithm runs in $O(n)$ time, independently of the alphabet size. Its space complexity is the space complexity of the DAWG, that is, linear with respect to the word length.

Section~\ref{sect-getting} describes a relatively simple algorithm to exhibit an attractor of a word. Since the question is NP-hard (see \cite{KempaP18}), the algorithm is based on a heuristics that works greedily to successively produce the positions of the attractor. Additionally, our algorithm runs in linear time when the attractor is bounded by a constant.

A general greedy quadratic algorithm to compute a small attractor is given in \cite[Definition 19]{SchaefferShallit2024}. 
(Another general method to approximate the construction of string attractors is to use Lempel-Ziv parsings known for their efficient compression ratios. This is done by Whittington in \cite{whittington2024} and also by Navarro et al. in \cite{NavarroOP2019}) where the size of attractors is related to Lempel-Ziv parsings.

The experiments carried out with our algorithm on linearly recurrent words produce, in most cases, attractors of constant size independently of the word length, as expected (see \cite{SchaefferShallit2024}).
Note that the hardness result does not necessarily apply to automatic sequences or to substitutive sequences generated by constant-length substitutions. This potentially opens the way to methods producing smallest attractors for these words.

Experiments on known families of words, carried over sufficiently long words, show that the output of our algorithm is often a small attractor, sometimes the smallest possible. An implementation is available on GitHub \cite{BealCrochemoreSoftware2025} that is able to deal with texts of a million letters.

Moreover, our algorithm can be modified to run in linear time independently of the attractor size, provided that a possible larger attractor is accepted.
A general method for finding an attractor of smallest size using MAX-SAT techniques is presented in
\cite{DBLP:conf/esa/Bannai0IKKN22}. Despite a greater running time, the implementation is able to deal with texts of about a million letters.

\section{Definitions}
\label{sect-definitions}
Let $\ttA=\{\tta,\ttb,\dots\}$ be an \emph{alphabet}, i.e., a finite set of letters.
A (finite) word $w$ over $\ttA$ is denoted by $w_1w_2\cdots w_{n}$ where all $w_i$ are in $\ttA$ and $n\ge 0$. Its \emph{length} $n$ is denoted by $|w|$ and the set of \emph{positions} on $w$ is $\llbracket 1, n \rrbracket$.
We denote the word $w_i w_{i+1} \cdots w_j$ by $w_{[i,j]}$, for $i,j\in\llbracket 1,n \rrbracket$. By convention, it is the empty word $\epsilon$ if $j < i$.
The word $w_{[i,j]}$ is called a \emph{factor} (or a \emph{block}) of $w$. It is a \emph{prefix} of $w$ when $i =1$ and a suffix of $w$ when $j = n$.
An \emph{occurrence} of a non-empty factor $u$ of $w$ is a pair $(i,j)$, $i \leq j$, of positions that satisfy $u = w_{[i,j]}$. The position $i$ is the starting position of the occurrence, and $j$ is its ending position.
Note that $j=i+|u|-1$. 

Let $\EP(u)$ denote the set of ending positions of all occurrences of $u$ in~$w$.
Let $\Suff(w)$ denote the set of suffixes of $w$, and let $\mathcal{S}(w)$ denote its Suffix automaton, that is, the (unique) minimal deterministic automaton accepting $\Suff(w)$. As such, the states of $\mathcal{S}(w)$ are (identified with) the equivalence classes of the set of factors of $w$ under the equivalence $\equiv$:
\begin{equation*}
u \equiv v \text{  if and only if  }  u^{-1} \Suff(w) = v^{-1} \Suff(w),
\end{equation*}
where $u^{-1}x = v$ if $x=uv$.
Equivalently, a state $q$ of $\mathcal{S}(w)$ is the set of factors of $w$ having the same set of ending positions. The smallest ending position of occurrences of any factor $u$ in $q$ is indicated by $\pos[q]$.

\smallskip
In this article, we consider all states of $\mathcal{S}(w)$ to be terminal states, which transforms the structure into a factor automaton. As such, it is also called a Directed Acyclic Word Graph (DAWG) after Blumer et al.
\cite{BlumerBHECS85} but is not necessarily a minimal automaton \cite{CHL07cup}.

The automaton $\mathcal{S}(w)$ is a tuple $(Q, I,\delta)$, where $Q$ is the set of states with $I$ the initial state as defined above and $\delta$ is the transition function defined accordingly.

\begin{figure}
\begin{center} 
    \begin{tikzpicture}[node distance=12mm,>=stealth] 
      \node[initial] (0) {$0$};
      \node[right of=0] (1) {$1$};
      \node[right of=1] (2) {$2$};
      \node[right of=2] (3) {$3$};
      \node[right of=3] (4) {$4$};
      \node[right of=4] (5) {$5$};
      \node[right of=5] (6) {$6$};
      \node[right of=6] (7) {$7$};
      \node[below of=3] (8) {$3''$};
      \node[below of=8] (9) {$3'$};
      \node[right of=9] (10) {$4''$};
      \draw (0) edge node {$\tta$} (1);
      \draw (1) edge node {$\tta$} (2);
      \draw (2) edge node {$\ttb$} (3);
      \draw (3) edge node {$\ttb$} (4);
      \draw (4) edge node {$\tta$} (5);
      \draw (5) edge node {$\ttb$} (6);
      \draw (6) edge node {$\ttb$} (7);
      \draw (1) edge[bend right=15] node {$\ttb$} (8);
      \draw (8) edge[bend left=10] node {$\ttb$} (10); 
      \draw (0) edge[bend right=10] node {$\ttb$} (9);
      \draw (10) edge[bend right=10] node[right] {$\tta$} (5);
      \draw (9) edge node {$\ttb$} (10);
      \draw (9) edge[bend right=10] node[right] {$\tta$} (5);
      \draw (1) edge[dashed, bend left]  (0);
      \draw (2) edge[dashed, bend left]  (1);
      \draw (3) edge[dashed] (8);
      \draw (4) edge[dashed] (10);
      \draw (5) edge[dashed, bend right=35] (1);
      \draw (6) edge[dashed, bend left=10] (8);
      \draw (7) edge[dashed, bend left=20] (10);
      \draw (8) edge[dashed] (9);
      \draw (10) edge[dashed, bend left] (9);
    \end{tikzpicture}
\end{center}
\caption{\label{figure.DAWG} The DAWG of $w=\mathtt{aabbabb}$. 
 Its initial state is $0$. Although useless in the paper, note that as a Suffix automaton its terminal states are $0$, $3'$, $4''$, and $7$ found on the suffix path from the last state to the initial state. 
 }
\end{figure}

Figure \ref{figure.DAWG} shows the DAWG of $w=\mathtt{aabbabb}$ with states (or nodes) and arcs. Plain arcs, labelled by letters occurring in the word, represent the transition function. The paths from the initial state $0$ are labelled by all the factors of $w$ and only them.
The dashed arcs of the DAWG correspond to a function $F \colon Q\setminus\{I\}\rightarrow Q$ called its failure or suffix function. They are called failure links or suffix links.
The function $F$ is defined as follows. Let $q$ be a state in $Q\setminus\{I\}$ 
and let $vu$ be any factor in $q$, i.e. $q=\delta(I,vu)$), where 
$u$ is the longest suffix of $vu$ with $\delta(I,u) = r \neq q$. Then $F[q]=r$, the class of factor $u$. 
The definition is independent of the choice of $vu$ in $q$.
The function $F$ is a crucial component both in the efficient construction of these automata and in the application of automata as pattern matching machines (see \cite[Chapters 5 and 6]{CHL07cup}).

A state $q$ of $\mathcal{S}(w)$ is called a \emph{prefix state} if $q = \delta(I, u)$ for a prefix $u$ of $w$.
We also say that a state $q$ of $\mathcal{S}(w)$ is a \emph{leaf} if there is no state $r$ for which $F(r) = q$.
Leaves are prefix states, but the converse is not true.
Non-prefix states are states that do not contain a prefix of the word. They are obtained as clones of other states during the incremental construction of $\mathcal{S}(w)$ (see \cite{BlumerBHECS85,CHL07cup}).

In addition to $\pos$ and $F$, the construction of $\mathcal{S}(w)$ also provides, for each state $q$, the maximum length $L[q]$ of the factors in $q$. 
The maximum length $L[q]$ is not used in our algorithm, but we need
the minimum length $\ell[q]$ of the factors in $q$. Its value is $0$ for $q=I$ and equals $L(F[q])+1$ when $q\neq I$. 

Here are the values of $F$ and $\ell$ for the example automaton $\mathcal{S}(\mathtt{aabbabb})$. The prefix states are the states $0, 1, 2, 3, 4, 5, 6, 7$.

\[
\begin{array}{|r|lllllllllll|}
\hline
j & & 1 & 2& 3 & & & 4 & & 5 & 6 & 7\\
\hline
w_j & & a & a & \ttb & & & \ttb & & a & \ttb & \ttb\\
\hline\hline
q       & 0 & 1 & 2 & 3 & 3'& 3''& 4 & 4''& 5 & 6 & 7 \\
\hline
F[q]    & - & 0 & 1 & 3''& 0 & 3'& 4''& 3'& 1 & 3''& 4''\\
\ell[q] & 0 & 1 & 2 & 3 & 1 & 2 & 4 & 2 & 2 & 3 & 4 \\
\pos[q] & 0 & 1 & 2 & 3 & 3 & 3 & 4 & 4 & 5 & 6 & 7 \\
\hline
\end{array}
\]

\section{Checking attractors}
\label{sect-checking}

In this section, we provide a linear-time algorithm to check whether a set of positions on a word is an attractor of the word.

Let $P$ be a set of positions on the word $w$ of length $n$, that is, a subset of $\llbracket 1, n \rrbracket$.
We say that a non-empty factor $u$ of $w$ \emph{covers} $P$ if there is an occurrence
$(i,j)$ of $u$ in $u$ and an element of $P$ in $\llbracket i, j \rrbracket$.
Recall that a set $P$ of positions on $w$ is a \emph{word attractor} of $w$ if all its factors cover $P$.

\begin{example}
Consider the word $w=\texttt{aabbabb}$ and the set $P = \{2,4\}$.
Clearly, every non-empty factor $u$ having an occurrence $(i,j)$ in $w$ with $i\leq2\leq j$ or $i\leq4\leq j$ covers $P$. In the following, we list all remaining occurrences of factors in $w$:
    \begin{itemize}
        \item $w_{[1,1]}=w_{[5,5]}=\tta$,
        \item $w_{[3,3]}=w_{[6,6]}=w_{[7,7]}=\ttb$, 
        \item $w_{[5,6]}=\texttt{ab}$,
        \item $w_{[6,7]}=\texttt{bb}$, and
        \item $w_{[5,7]}=\texttt{abb}$.
    \end{itemize}
    Since all these factors cover $P$ through their occurrences $(2,2)$, $(4,4)$, $(2,3)$, $(3,4)$, and $(2,4)$, respectively, all factors of $w$ cover $P$ and the set $P$ is a word attractor of $w$.
\end{example}

The definition can be reformulated in terms of the Suffix automaton of the word as follows. 

\begin{lemma} \label{lemma.smallest}
A set $P$ of positions on a word $w$ is a word attractor of $w$ if and only if, for any state $q$, $q\neq I$, of $\mathcal{S}(w)$, the shortest word $u$ for which $\delta(I,u)=q$ covers $P$.
\end{lemma}
\begin{proof}
The automaton $\mathcal{S}(w)$ accepts all factors occurring in $w$ since all states are terminal: for every non-empty factor $u$ of $w$, there exists a state $q \neq I$ of $\mathcal{S}(w)$ with $\delta(I,u)=q$. Thus, if $P$ is a word attractor, $u$ covers $P$. 

Conversely, let $v$ be a factor of $w$ and $\delta(I,v)=q$. Let $u$ be the shortest factor of $w$ such that $\delta(I,u)=q$. If $u$ covers $P$, then there is an occurrence $(i, j)$ of $u$ in $w$ that contains a position $k$ in $P$. Thus, there is an occurrence $(i', j)$ of $v$ for some $i' \leq i$ of $v$ in $w$ and this occurrence contains $k$.
\end{proof}

Let $P$ be a set of positions on $w$. 
For a position $j$ on $w$, we define the \emph{distance from $j$ to $P$ on the left} by 
\begin{equation*}
d_P(j) = \min\{j-i \mid i \in P,\; i \leq j\}\cup\{+\infty\}.
\end{equation*}
For a factor $u$ of $w$, we define
\begin{equation*}
d_P(u) = \min\{ d_P(j) \mid j \in \EP(u)\}.
\end{equation*}
and, for a state $q$ of $\mathcal{S}(w)$,
\begin{equation*}
d_P(q) = d_P(u), 
\end{equation*}
where $u$ is the shortest word in $q$.

The lemma~\ref{lemma.smallest} leads to the main property on which the algorithm \Algo{IsAttractor} below is built.

\begin{lemma}  \label{lemma.distance}
A set $P$ of positions on a word $w$ is a word attractor of $w$ if and only if, for each state $q$, $q\neq I$, of the Suffix automaton of $w$, $d_P(q)<\ell[q]$.
\end{lemma}
\begin{proof}
This is a direct consequence of Lemma~\ref{lemma.smallest} and of the definition of $d_P$. 
\end{proof}

The algorithm \Algo{IsAttractor} checks whether a set $P$ of positions on a word $w$ is a word attractor. In the first step, it computes the values $d(j)$ representing $d_P(j)$, for all positions $j$ on $w$. 
We assume that we have the table $T$ of prefix states, defined by $T[j]=q$ for the prefix state $q = \delta(I, w_{[1, j]})$ with $j\in\llbracket 1,|w|\rrbracket$.
Then, $d[q]$ is initialised to $d[T[j]]$. 
If \( q \) is not a prefix state, \( d[q] \) is initialised to \( +\infty \).

The second step computes the values $d[q]$ according to the definition of \( d_P(q) \) during a bottom-up traversal of the tree of suffix links. Specifically, the function \( F \) infers a parent link between states, which defines the tree.

At the end of both steps of the computation, the values $d[q]$ represent $d_P(q)$ for all states $q$ of the automaton. 

Figure~\ref{figure.links} displays the tree of suffix links associated with the previous example. 

\begin{figure}
\begin{center} 
    \begin{tikzpicture}[node distance=12mm,>=stealth] 
      \node[initial] (0) {$0$};
      \node[right of=0] (1) {$1$};
      \node[right of=1] (2) {$2$};
      \node[right of=2] (3) {$3$};
      \node[right of=3] (4) {$4$};
      \node[right of=4] (5) {$5$};
      \node[right of=5] (6) {$6$};
      \node[right of=6] (7) {$7$};
      \node[below of=3] (8) {$3''$};
      \node[below of=8] (9) {$3'$};
      \node[right of=9] (10) {$4''$};
      \path (.63,.2) node (51) {$\tta$};
      \node[right of=51] (52) {$\tta$};
      \node[right of=52] (53) {$\ttb$};
      \node[right of=53] (54) {$\ttb$};
      \node[right of=54] (55) {$\tta$};
      \node[right of=55] (56) {$\ttb$};
      \node[right of=56] (57) {$\ttb$};
      \draw (1) edge[dashed, bend left]  (0);
      \draw (2) edge[dashed, bend left]  (1);
      \draw (3) edge[dashed] (8);
      \draw (4) edge[dashed] (10);
      \draw (5) edge[dashed, bend right=30] (1);
      \draw (6) edge[dashed, bend left=10] (8);
      \draw (7) edge[dashed, bend left=20] (10);
      \draw (8) edge[dashed] (9);
      \draw (9) edge[dashed, bend left] (0);
      \draw (10) edge[dashed, bend left] (9);
    \end{tikzpicture}
\end{center}
\caption{\label{figure.links} The tree of suffix links of the DAWG $\mathcal{S(\mathtt{aabbabb})}$ in Figure~\ref{figure.DAWG}. 
Its root is the initial state $0$, its prefix states are states 0,1,2,3,4,5,6,7, and its leaves are states $2, 3, 4, 5, 6$, and $7$.}
\end{figure}

Eventually, the algorithm stops during the execution of the for loop in lines \ref{algo1-line-10}-\ref{algo1-line-13} if ever the condition in Lemma~\ref{lemma.distance} is not satisfied, or successfully after the full execution of the loop.

\begin{algo}{IsAttractor}{\mbox{word }w, P \mbox{ set of positions on }w}
\COM{$\mathcal{T}$ is the tree of suffix links of the DAWG $\mathcal{S}(w)$}
\COM{$F$ is the suffix links table}
\COM{$T$ is the table of prefix states}  \RCOM{7}{$T[i] = \delta(I,w_{[1,i]})$}
  \DOFOR"{each state $q\neq I$ of $\mathcal{S}(w)$}
    \SET{d[q]}{\infty}
  \OD
  \SET{k}{-\infty}
  \DOFORI{j}{1}{|w|}
    \IF{j\in P}
      \SET{k}{j}
    \FI
    \SET{d[T[j]]}{j-k}
  \OD 
  \DOFOR"{each node $q\neq I$ in a bottom-up traversal of $\mathcal{T}$} \label{algo1-line-10}
      \SET{d[q]}{\min\{d[q]\}\cup\{d[r] \mid F(r)=q\}} \label{algo1-line-11}
    \IF{d[q]\geq\ell[q]}  \label{algo1-line-12}
      \RETURN \False \label{algo1-line-13}
    \FI
  \OD 
  \RETURN \True \label{algo1-line-14}
\end{algo}

\begin{proposition} 
The algorithm \Algo{IsAttractor} checks whether a set $P$ of positions on a word $w$ is a word attractor of $w$. That is, $P$ is a word attractor of $w$ if and only if $\Call{IsAttractor}{w,P}=\True$. 
\end{proposition}
\begin{proof}
Let $\mathcal{S}$ be the DAWG $\mathcal{S}(w)$, and $\mathcal{T}$ the associated tree of suffix links.
We show that $d[q] = d_P(q)$ for all states $q\neq I$ of $\mathcal{S}$
when \Algo{IsAttractor} stops.
At the end of the execution, for each state $q \neq I$, one has
\[
d[q] = \min \{d(j) \mid j\in\llbracket 1,|w|\rrbracket\; \exists i \geq 0\; F^i[r] = q \text{, where } r= \delta(I, w_{[1,j]})\}.
\]
Note that, if $q$ is a leaf of $\mathcal{T}$, then it is also a prefix state $\delta(I, w_{[1,j]})$ and $d[q] = d(j)$,
at the end of the execution.

By definition, we have
\[
d_P(q) = \min \{d_P(j) \mid j \in \EP(q)\}.
\]
For $j\in\llbracket 1,|w|\rrbracket$, let $r=\delta(I, w_{[1, j]})$ (then, $\pos[r_j] = j$).
For each factor $u$ for which $\delta(I, u) = q$, 
$u$ is a suffix of $w_{[1, j]}$ for each $j \in \EP(u)$.
Thus, for each $j \in \EP(u)$, there is a nonnegative integer $i$ such that $F^{i}[r] = q$.
Conversely, if $r$ is a prefix state $\delta(I, w_{[1, j]})$ such that $F^i[r] = q$ for some $i \geq 0$,
then $\EP(w_{[1,j]}) \subseteq \EP(u)$ for each factor $u$ of $w$ with $\delta(I, u) = q$.
This implies that $u$ is a suffix of $w_{[1,j]}$.
Thus, 
\[
d_P(q) = \min \{d_P(j) \mid \exists i \geq 0 \text{ such that } F^i[r] = q \text{ where } r = \delta(I, w_{[1,j]})\}.
\]
As a consequence, at the end of the execution, $d[q] = d_P(q)$ for each state $q \neq I$. 
The conclusion follows from Lemma \ref{lemma.distance}.
\end{proof}

\begin{proposition}
The algorithm \Algo{IsAttractor} can be implemented to run in linear time with respect to $|w|$. If the Suffix automaton of $w$ is available, this time is independent of the alphabet size.
\end{proposition}
\begin{proof}
For a word $w$ of length $n$ on a fixed size alphabet, the DAWG $\mathcal{S}(w)$, together with the tables $F$, $\ell$ and $T$, can be computed 
in time $O(n)$ \cite{BlumerBHECS85} (see also \cite{CHL07cup}).
It becomes $O(n \log |A|)$ for a general ordered alphabet. It is done in $O(n)$ in \cite{FujishigeTIBT23} for an integer alphabet of polynomial size. 
The $O(n)$ time complexity can also be reached for integer alphabets of polynomial size 
using a sparse list implementation (see, for example, \cite[Exercise 2.12 p. 71]{AhoHU74}).

This running time after the computation of the DAWG and its attributes is $O(n)$.
It is clear after noting that both:
\begin{itemize}\topsep0pt\parskip0pt
\item each suffix link is traversed only once at line~\ref{algo1-line-11},
\item there are at most $2|w|-2$ suffix links since the number of states of $\mathcal{S}(w)$ is at most $2|w|-1$ (see \cite[page 200]{CHL07cup}).
\end{itemize}
Thus, the total remaining running time of \Call{IsAttractor}{w}, after the computation of the DAWG, is $O(n)$ independently of the alphabet size. The space complexity is the same as the space complexity of the DAWG.
\end{proof}

\section{Greedy attractor}
\label{sect-getting}

Finding a smallest string attractor for a word is known to be NP-complete \cite{KempaP18}. 
Thus, it is natural to consider approximations and heuristic algorithms. 
A general greedy quadratic algorithm to compute a small attractor is given in \cite[Definition 19]{SchaefferShallit2024}. 
In this section, we design a greedy heuristic algorithm to find a small attractor for a given word $w$. Applied to some families of words, the algorithm outputs very small attractors and sometimes even the smallest. 
The running time is linear with respect to the length of the word if the size of the attractor is bounded by a constant, that is, of size $O(1)$. 
When the attractor is potentially of arbitrary size, the algorithm can be changed to run in linear time with a simple technical improvement (see end of section), but this may affect the final size of the attractor found.

The main idea to find a small attractor of a word $w$ is to use the DAWG of $w$ and, in some sense, to reverse the process of Algorithm \Algo{IsAttractor}.

\subsection*{Informal presentation}
The algorithm's design employs an approach based on certain observations regarding the properties of the DAWG $\mathcal{S}(w)$.

\paragraph{Note 1}
To get an attractor $P$ for $w$ with its DAWG $\mathcal{S}(w)$, for each leaf $j$ of the Suffix links tree $\mathcal{T}$, one must have $d_P(j) < \ell[j]$.
Indeed, there must be a position $i\in P$ with both $i\leq j$ and $j-i<\ell[j]$. 
Several leaves can share the same such position $i$.

For instance, for the leaf $6$ of $\mathcal{T}$ in Figure \ref{figure.links}, 
at least one of the positions $4, 5, 6$ should be in $P$. 
In fact, for this example, the unique occurrence of the factor $\texttt{bab}$ of $\texttt{aabbabb}$ covers only these positions.

\paragraph{Note 2}
The indegree of states in $\mathcal{T}$ that are not prefix states is at least $2$. Actually, during the construction of $\mathcal{S}(w)$, when a state $q$ is cloned into the state $q'$, the latter gets two suffix links pointing to it. The clone state $q'$ is an internal node of $\mathcal{T}$, and all clone states are accessible by suffix paths from leaves.

For instance, in Figure \ref{figure.links}, the states $3'$, $3''$, and $4''$ have indegree $2$ in $\mathcal{T}$.

\paragraph{Heuristics}
The heuristics consists in selecting successive positions on the word that are candidates for being in the attractor. As explained in the following, this is done greedily by scanning the word from the end to the beginning of the word. 

For the running example of the word $\texttt{aabbabb}$,
the drawing below displays the values $\ell[q]$ for all states of the tree $\mathcal{T}$ in Figure~\ref{figure.links}. The length of each horizontal bar replicates the value $\ell[q]$ associated with a state $q$. The bar itself represents the shortest word $u$, where $|u|=\ell[q]$, that satisfies $\delta(I,u)=q$ in the DAWG. The ending position of the bar is $\pos[q]$.

\hspace{.9mm}
    \begin{tikzpicture}[node distance=12mm,>=stealth] 
      \path (0.6,.2) node (51) {$\tta$};
      \node[right of=51] (52) {$\tta$};
      \node[right of=52] (53) {$\ttb$};
      \node[right of=53] (54) {$\ttb$};
      \node[right of=54] (55) {$\tta$};
      \node[right of=55] (56) {$\ttb$};
      \node[right of=56] (57) {$\ttb$};
      \node[initial] (0) {$0$};
      \node[right of=0] (1) {$1$};
      \node[right of=1] (2) {$2$};
      \node[right of=2] (3) {$3$};
      \node[right of=3] (4) {$4$};
      \node[right of=4] (5) {$5$};
      \node[right of=5] (6) {$6$};
      \node[right of=6] (7) {$7$};
      \node[below of=3] (8) {$3''$};
      \node[below of=8] (9) {$3'$};
      \node[right of=9] (10) {$4''$};
    \path (0,-.5) node (100) {($\ell[q]$)}; 
    \node[right of=100] (101) {($1$)}; 
    \node[right of=101] (102) {($2$)}; 
    \node[right of=102] (103) {($3$)}; 
    \node[right of=103] (104) {($4$)}; 
    \node[right of=104] (105) {($2$)}; 
    \node[right of=105] (106) {($3$)}; 
    \node[right of=106] (107) {($4$)}; 
    \node[below of=103] (108) {($2$)}; 
    \node[below of=108] (109) {($1$)}; 
    \node[right of=109] (110) {($2$)}; 
    \draw (0.3,-3.2) -- (1.5,-3.2); 
    \draw (0.3,-3.4) -- (2.7,-3.4); 
    \draw (0.3,-3.6) -- (3.9,-3.6); 
    \draw (1.5,-3.8) -- (3.9,-3.8); 
    \draw (2.7,-4.0) -- (3.9,-4.0); 
    \draw (0.3,-4.2) -- (5.1,-4.2); 
    \draw (2.7,-4.4) -- (5.1,-4.4); 
    \draw (3.9,-4.6) -- (6.3,-4.6); 
    \draw (3.9,-4.8) -- (7.5,-4.8); 
    \draw (3.9,-5.0) -- (8.7,-5.0); 
\end{tikzpicture}

\medskip
The heuristics proceeds by scanning the word backward in a certain number of steps $j_1 > j_2 > \cdots > j_m$ corresponding to some positions on the word.
For each step $j$, the heuristics consists in choosing
the smallest position $k$, $0 \leq k \leq j$, for which each position $i$ with $k \leq i \leq j$, is covered by at least one bar covering $k$ as well. Informally, this means that $k$ is covered by as many bars as possible among those that do not yet cover a position that has already been selected.
Choosing more bars for each selected position globally reduces the number of positions to be chosen from, as well as the eventual size of the final attractor. 

In the running example, it is clear that, when $j=7$, the first selected position is ideally position $4$, covered by $5$ bars. This implies that the condition in Lemma~\ref{lemma.distance} on distances associated with states $4$, $5$, $6$, and $7$ will hold true. 

The situation is more interesting when dealing with the prefix $\texttt{aab}$. Looking roughly at the remaining bars, position $3$ could be chosen because $\ell[3']=1$, which eventually leads to an attractor of size $3$. However, it is already known that the condition holds on state $3'$ (indeed, the position $3$ is covered only by the letter $\texttt{b}$ whose next occurrence covers position $4$ of the attractor). Therefore, this leads us to discard the bar associated with states $3'$ and to choose position $2$ as the next position of the complete attractor instead. This is realised by Algorithm \Algo{GetAnAttractor} below and shown on the run of it with the DAWG of the above example word as input.

\subsection*{Algorithm description}

The algorithm \Algo{GetAnAttractor} computes a sequence of contiguous intervals $(I_r)$, $r=1,2,\dots,m$, where
$I_r = \llbracket k_r , j_r \rrbracket$, and $k_r, j_r$ are positions on $w$ such that $j_1=|w|$ and $j_{r+1} = k_r - 1$, for $r>0$. 
It also computes a sequence of arrays $(d_r)$, $r=0,1,2,\dots,m$, indexed by the set of states of $\DAWG(w)$, that satisfies
\begin{equation*}
k_r = \max\{\pos[q] -\ell[q] + 1 \mid \pos[q] \in I_r \text{ and } d_{r-1}[q] \geq \ell[q]\},
\end{equation*}
after setting $d_0(q)=+\infty$ for each state $q$. (The inequality refers to Lemma~\ref{lemma.distance}.)
The algorithm outputs the sequence $(k_r)$, $r=1,2,\dots,m$, as the found attractor.

The algorithm proceeds by performing three main steps in each iteration $r$.
The first step, called \textit{stacking}, consists in computing $k_r$ for the current position $j_r$ as suggested above. The second step adds it to the future attractor. The third step consists in \textit{ updating} $d$, the table of distances on states. 
This is done as follows: the distances between $k_r$ and all the positions in $I_r = \llbracket k_r,j_r \rrbracket$ are calculated and propagated to all concerned states using their suffix links, as done in the algorithm of Section~\ref{sect-checking}.
The process resumes at position $k_r-1$, which becomes the value $j_{r+1}$.
The distance table obtained at the end of the iteration $r$ is called~$d_r$.

\bigskip
In Algorithm \Algo{GetAnAttractor} below, after initialisations on lines~\ref{algo6-line-1}-\ref{algo6-line-4}, the iteration controlled by the variable $j$ executes the three steps, stacking in lines~\ref{algo6-line-6}-\ref{algo6-line-12}, adding in lines~\ref{algo6-line-13} and updating in lines~\ref{algo6-line-14}-\ref{algo6-line-22}.

The whole run of \Call{GetAnAttractor}{\mathcal{S}(w)} is driven by a traversal of the prefix states of $\mathcal{S}(w)$ done backward, that is, from the last state to the initial state, inside the main while loop. 

Note 1 is applied in lines~\ref{algo6-line-6}-\ref{algo6-line-12} to find the next position to be added to the future attractor. The choice of the value $k$ ensures that the condition $d[q]<\ell[q]$ will be satisfied for all the considered prefix states.

Finally, in lines~\ref{algo6-line-14}-\ref{algo6-line-22}, following Note 2, the table $d$ is partially completed both for the considered prefix states and for their parents in the suffix tree using the suffix link table $F$.

Meanwhile, the algorithm stores the states $q$ for which the condition $d[q]<\ell[q]$ is met (variable $\DONE$), and stops if this condition holds for all states. Then, it outputs the set $P$ of positions, which is an attractor of the word as shown in Proposition \ref{proposition.correction}.

The input of Algorithm \Algo{GetAnAttractor} is the DAWG $\mathcal{S}(w)$, which is given with these elements:
\begin{itemize}\topsep0mm\parskip0mm
\item
$\mathcal{T}$ its tree of suffix links; 
\item
$F$ the suffix links table (parent links of the tree);
\item
$T$ is the table of prefix states  ($T[i] = \delta(I,w_{[1,i]})$);
\item
$\pos$ position table of states ($\pos[q]$ is the smallest end position of any  factor $u$ of $w$ with $q=\delta(I,u)$).
\end{itemize}

\begin{algo}{GetAnAttractor}{\mathcal{S}(w)}
 \DOFOR"{all states $q$ of $\mathcal{S}(w)$} \label{algo6-line-1}
   \SET{d[q]}{\infty}
 \OD
 \SET{(P,\DONE,j)}{(\emptyset,\{I\},|w|)} \label{algo6-line-4}
 \DOWHILE{j > 0}
   \SET{(i,k)}{(j,1)} \RCOM{6}{stacking down from $j$} \label{algo6-line-6}
   \DOWHILE{k \leq i}                        \label{algo6-line-7}
     \SET{q}{T[i]}
     \DOWHILE{\pos[q]=i}
       \IF{d[q]\geq\ell[q]}  \label{algo6-line-9}
         \SET{k}{\max\{k,i-\ell[q]+1\}}
       \FI
       \SET{q}{F[q]} \label{algo6-line-11}
     \OD
     \SET{i}{i-1}  \label{algo6-line-12}
   \OD   
   \SET{P}{P \cup \{k\}} \RCOM{6}{adding to attractor} \label{algo6-line-13}
   \DOFORI{i}{k}{j}  \RCOM{6}{updating $d[q]$'s upward from $q=T[i]$}      \label{algo6-line-14}
     \SET{(e,q)}{(i-k,T[i])}
     \DOWHILE{q \neq I \mbox{ and } e < d[q]}    \label{algo6-line-16}
               \SET{d[q]}{e}                     \label{algo6-line-17}
       \IF{d[q]<\ell[q]}              
         \SET{\DONE}{\DONE\cup\{q\}}
         \IF{|\DONE|=|\mathcal{S}(w)|}            \label{algo6-line-19}
            \RETURN{P} 
         \FI
       \FI
       \SET{q}{F(q)}        \label{algo6-line-22}
     \OD
        \OD
   \SET{j}{k-1}           \label{algo3-line-23}
 \OD
\end{algo}

\begin{proposition} \label{proposition.correction}
The algorithm \Algo{GetAnAttractor} applied to the DAWG $\mathcal{S}(w)$ of the word $w$ outputs an attractor of $w$.
\end{proposition}
\begin{proof}
Before the first step in the main while loop, the initial distance table $d_0$ is set, for each state $q$, as $d_0[q] = + \infty$.
At the step $1$, $j_1= |w|$.
For each step $r$, 
\begin{equation*}
k_r = \max\{\pos[q] -\ell[q] + 1 \mid \pos[q] \in I_r \text{ and } d_{r-1}[q] \geq \ell[q]\},
\end{equation*}
where $d_r$ is the value of table $d$ at the end of step $r$.

Let $P_r = (k_r, k_{r-1}, \ldots, k_1)$ be the set of positions already computed at the end of step $r$ and stored in $P$.
We say that a state $q$ is $r$-checked if the shortest factor of $q$ covers a position in $P_r$.

Let us show the following loop invariants. For each $r$, $1\leq r \leq m$,
\begin{itemize}
\item $d_r[q] = \min \{d_{P_r}(u) \mid u \text{ in } q\}$,
where $d_{P_r}(u) = \min\{ d_{P_r}(j) \mid j \in \EP(u)\}$;
\item for each state $q$ for which $\pos(q) \geq k_r$, $q$ is $r$-checked.
\end{itemize}
Assume by induction that the invariants hold for all steps less than $r$.

For each state $q$ whose position satisfies $\pos[q] \geq k_r$, either 
$\pos[q] \geq k_{r-1}$, and $q$ is $(r-1)$-checked, and thus $r$-checked, or
$\pos[q] \in I_r$. In the latter case, by definition of $k_r$,  
either $d_{r-1}[q] < \ell[q]$ or $\pos[q] -\ell[q] + 1 \leq k_r$.
If $d_{r-1}[q] < \ell[q]$, then $q$ is $(r-1)$-checked, and thus $r$-checked.
If $\pos[q] -\ell[q] + 1 \leq k_r$, then $|k_r - \pos[q]| < \ell[q]$,
and thus $q$ is $r$-checked. 

Assume by induction again that $d_{r-1}[q] = \min \{d_{P_{r-1}}(u) \mid u \text{ in } q\}$.
After the update of the distance table, at the end of step $r$, we get
$d_{r}[q] = \min \{d_{r-1}[q]\} \cup \{d_{k_r}(j) \mid j \in \EP(u) \cap I_r\}$.
Hence, $d_{r}[q] = \min \{d_{P_r}(u) \mid u \text{ in } q\}$.

Finally, if the algorithm stops at step $m$, all states $q$ are $m$-checked, which implies that $P_m$ is an attractor.
\end{proof}

\begin{proposition}  \label{proposition.complexity}
The execution of \Call{GetAnAttractor}{\mathcal{S}(w)} that produces a word attractor $P$ runs in time $O(|w|\cdot|P|)$.
\end{proposition}
\begin{proof}
First, recall that the DAWG size is linear with respect to $|w|$. More precisely, the total number of suffix links in $\mathcal{S}(w)$ is less than $2|w|-1$ (see the size of a suffix automaton in \cite[page 200]{CHL07cup}).

To prove the expected running time, we evaluate the total runtime of all the executions of the while loop in  lines~\ref{algo6-line-7}-\ref{algo6-line-12} and the total runtime of all the executions of the for loop in lines~\ref{algo6-line-14}-\ref{algo6-line-22}.

For the while loop, 
notice that, at step $r$, the position $k_r$ is searched in the interval
$\llbracket j -\ell[T[j]] + 1, j \rrbracket$, 
and the intervals $I_r = \llbracket k_r , j_r \rrbracket$ are disjoint.
Therefore, the total runtime is $O(|w|)$ because the suffix links of $\mathcal{T}$ considered in line~\ref{algo6-line-11} are traversed only once.

As for the for loop, note that the values of $e$ for a given pair $(k,j)$ are successively $0$, $1$, \dots, $j-k$ (in increasing order). Therefore, if for some state $q$, $d[q]$ is updated to $e$, it will not be updated again with a next value of $e$ because the condition $e < d[q]$ will not hold, and the while loop stops. Overall, for a given pair $(k,j)$, this step is executed in linear time. 

However, $d[q]$ may be updated again in a subsequent step for another pair $(k,j)$.
Since each pair is associated with an element of the attractor, the total running time is then $O(|w| \cdot|P|)$.
\end{proof}

\subsection*{Linear-time version}
Experiments suggest show that the running time of the computation is close to linear with respect to the length of the word. Nevertheless, a simple change in \Algo{GetAnAttractor} could make it work in linear time: count accesses to state $q$ in line~\ref{algo6-line-17} and executes instruction in line~\ref{algo6-line-22} only if the number of accesses to $q$ equals its indegree in the tree $\mathcal{T}$; else change $q$ to $I$ to stop the while loop in lines~\ref{algo6-line-16}-\ref{algo6-line-22}.
This way, the suffix link from $q$ to $F[q]$, for any $q$, is followed only when necessary and only once, leading to a linear-time version. However, the change possibly results in a larger attractor.

\subsection*{Examples of outputs}

A simple verification shows that \Call{GetAnAttractor}{\mathcal{S}(\texttt{a}^n)}, $n>0$, outputs the attractor $\{1\}$ as expected. \Call{GetAnAttractor}{\mathcal{S}(\texttt{ba}^n)} outputs the smallest attractor $\{1,2\}$ as expected as well. The next example also shows the good behavior of the algorithm to exhibit a small attractor.

\subsubsection*{Running the example word}

Below is an execution of \Algo{GetAnAttractor} on the running example word \texttt{aabbabb} using the tree of suffix links (Figure~\ref{figure.links}) of its DAWG (Figure~\ref{figure.DAWG}).

The nodes of the suffix links tree are labeled by pairs $(\ell,d)$. Initially,
the potential attractor is empty, $P = \emptyset$, and values $d[q]$ are set to $\infty$ for all states $q\neq I$.

\medskip
$P = \emptyset$. Initialisation of pairs $(\ell,d)$.

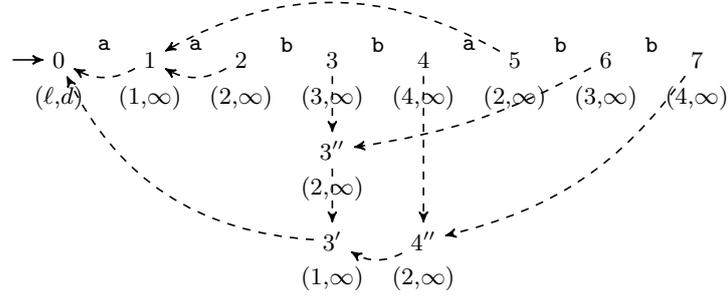
    \begin{tikzpicture}[node distance=12mm,>=stealth] 
      \path (.6,.2) node (51) {$\tta$};
      \node[right of=51] (52) {$\tta$};
      \node[right of=52] (53) {$\ttb$};
      \node[right of=53] (54) {$\ttb$};
      \node[right of=54] (55) {$\tta$};
      \node[right of=55] (56) {$\ttb$};
      \node[right of=56] (57) {$\ttb$};
      \node[initial] (0) {$0$};
      \node[right of=0] (1) {$1$};
      \node[right of=1] (2) {$2$};
      \node[right of=2] (3) {$3$};
      \node[right of=3] (4) {$4$};
      \node[right of=4] (5) {$5$};
      \node[right of=5] (6) {$6$};
      \node[right of=6] (7) {$7$};
      \node[below of=3] (8) {$3''$};
      \node[below of=8] (9) {$3'$};
      \node[right of=9] (10) {$4''$};
      \draw (1) edge[dashed, bend left]  (0);
      \draw (2) edge[dashed, bend left]  (1);
      \draw (3) edge[dashed] (8);
      \draw (4) edge[dashed] (10);
      \draw (5) edge[dashed, bend right=30] (1);
      \draw (6) edge[dashed, bend left=10] (8);
      \draw (7) edge[dashed, bend left=20] (10);
      \draw (8) edge[dashed] (9);
      \draw (9) edge[dashed, bend left] (0);
      \draw (10) edge[dashed, bend left] (9);
    \path (0,-.5) node (100) {($\ell$,$d$)};
    \node[right of=100] (101) {(1,$\infty$)};
    \node[right of=101] (102) {(2,$\infty$)};
    \node[right of=102] (103) {(3,$\infty$)};
    \node[right of=103] (104) {(4,$\infty$)};
    \node[right of=104] (105) {(2,$\infty$)};
    \node[right of=105] (106) {(3,$\infty$)};
    \node[right of=106] (107) {(4,$\infty$)};
    \node[below of=103] (108) {(2,$\infty$)};
    \node[below of=108] (109) {(1,$\infty$)};
    \node[right of=109] (110) {(2,$\infty$)};
    \end{tikzpicture}

$P = \{4\}$. Update of values $d[q]$ when position 4 is added to $P$

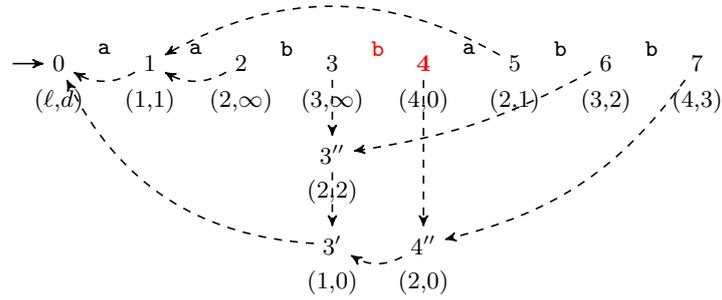
    \begin{tikzpicture}[node distance=12mm,>=stealth] 
      \path (.6,.2) node (51) {$\tta$};
      \node[right of=51] (52) {$\tta$};
      \node[right of=52] (53) {$\ttb$};
      \node[right of=53] (54) {$\textcolor{red}{\ttb}$};
      \node[right of=54] (55) {$\tta$};
      \node[right of=55] (56) {$\ttb$};
      \node[right of=56] (57) {$\ttb$};
      \node[initial] (0) {$0$};
      \node[right of=0] (1) {$1$};
      \node[right of=1] (2) {$2$};
      \node[right of=2] (3) {$3$};
      \node[right of=3] (4) {$\textcolor{red}{\mathbf{4}}$};
      \node[right of=4] (5) {$5$};
      \node[right of=5] (6) {$6$};
      \node[right of=6] (7) {$7$};
      \node[below of=3] (8) {$3''$};
      \node[below of=8] (9) {$3'$};
      \node[right of=9] (10) {$4''$};
      \draw (1) edge[dashed, bend left]  (0);
      \draw (2) edge[dashed, bend left]  (1);
      \draw (3) edge[dashed] (8);
      \draw (4) edge[dashed] (10);
      \draw (5) edge[dashed, bend right=30] (1);
      \draw (6) edge[dashed, bend left=10] (8);
      \draw (7) edge[dashed, bend left=20] (10);
      \draw (8) edge[dashed] (9);
      \draw (9) edge[dashed, bend left] (0);
      \draw (10) edge[dashed, bend left] (9);
    \path (0,-.5) node (100) {($\ell$,$d$)};
    \node[right of=100] (101) {(1,1)};
    \node[right of=101] (102) {(2,$\infty$)};
    \node[right of=102] (103) {(3,$\infty$)};
    \node[right of=103] (104) {(4,0)};
    \node[right of=104] (105) {(2,1)};
    \node[right of=105] (106) {(3,2)};
    \node[right of=106] (107) {(4,3)};
    \node[below of=103] (108) {(2,2)};
    \node[below of=108] (109) {(1,0)};
    \node[right of=109] (110) {(2,0)};
    \end{tikzpicture}

\noindent
Notice the role of the test in line ~\ref{algo6-line-9}: it enables state $3'$ to be discarded from the search for the next value, $k=2$. Without this test, the value of $k$ would be $3$ at the end of the loop, and another position would need to be added, since state 2 would still be unchecked.

\smallskip
$P = \{2,4\}$. Update of values $d[q]$ when position 2 is added to $P$.

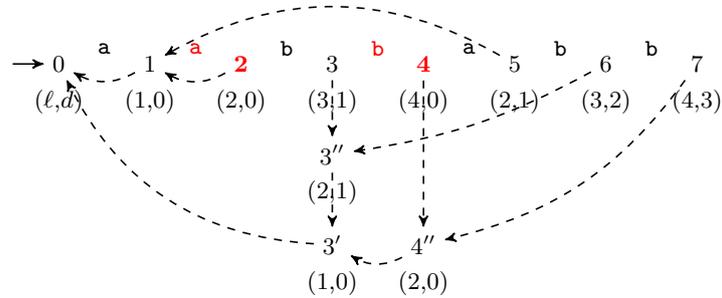
    \begin{tikzpicture}[node distance=12mm,>=stealth] 
      \path (.6,.2) node (51) {$\tta$};
      \node[right of=51] (52) {$\textcolor{red}{\tta}$};
      \node[right of=52] (53) {$\ttb$};
      \node[right of=53] (54) {$\textcolor{red}{\ttb}$};
      \node[right of=54] (55) {$\tta$};
      \node[right of=55] (56) {$\ttb$};
      \node[right of=56] (57) {$\ttb$};
      \node[initial] (0) {$0$};
      \node[right of=0] (1) {$1$};
      \node[right of=1] (2) {$\textcolor{red}{\mathbf{2}}$};
      \node[right of=2] (3) {$3$};
      \node[right of=3] (4) {$\textcolor{red}{\mathbf{4}}$};
      \node[right of=4] (5) {$5$};
      \node[right of=5] (6) {$6$};
      \node[right of=6] (7) {$7$};
      \node[below of=3] (8) {$3''$};
      \node[below of=8] (9) {$3'$};
      \node[right of=9] (10) {$4''$};
      \draw (1) edge[dashed, bend left]  (0);
      \draw (2) edge[dashed, bend left]  (1);
      \draw (3) edge[dashed] (8);
      \draw (4) edge[dashed] (10);
      \draw (5) edge[dashed, bend right=30] (1);
      \draw (6) edge[dashed, bend left=10] (8);
      \draw (7) edge[dashed, bend left=20] (10);
      \draw (8) edge[dashed] (9);
      \draw (9) edge[dashed, bend left] (0);
      \draw (10) edge[dashed, bend left] (9);
    \path (0,-.5) node (100) {($\ell$,$d$)};
    \node[right of=100] (101) {(1,0)};
    \node[right of=101] (102) {(2,0)};
    \node[right of=102] (103) {(3,1)};
    \node[right of=103] (104) {(4,0)};
    \node[right of=104] (105) {(2,1)};
    \node[right of=105] (106) {(3,2)};
    \node[right of=106] (107) {(4,3)};
    \node[below of=103] (108) {(2,1)};
    \node[below of=108] (109) {(1,0)};
    \node[right of=109] (110) {(2,0)};
    \end{tikzpicture}

\noindent
Since the condition $d[q]<\ell[q]$ holds for all states $q\neq I$, $\{2,4\}$ is an attractor, obviously of minimum size.
\section{Experimental results for known families of sequences}

This section provides experimental results for some known families of sequences. In addition, it further explores and extends the tests in \cite{SchaefferShallit2024} with experiments carried out on much longer sequences. The implementation of Algorithm \Algo{GetAnAttractor} used to carry out the experiments is available on GitHub \cite{BealCrochemoreSoftware2025}.

A \emph{substitutive sequence} is a sequence $\sigma^\omega(\tta)$, where 
$\sigma\colon A^* \to A^*$ is a morphism, $\tta \in A$, and $\sigma(\tta) = \tta u$, for a nonempty word $u$. 
Then, $\lim_{k \to + \infty} |\sigma^k(\tta)| = + \infty$. 

The \emph{Fibonacci sequence} is obtained with the morphism
$\sigma\colon \tta \to \tta\ttb, \ttb \to \tta$.
A smallest attractor has a size of 2 for a sufficiently long word \cite{MantaciRRRS21} (see also \cite{CrochemoreLR25addpbs}).

The \emph{Thue-Morse sequence} is obtained with the morphism
$\sigma\colon \tta \to \tta\ttb, \ttb \to \ttb\tta$.
A smallest attractor has a size of 4 \cite{MantaciRRRS21} for a sufficiently long word (see also \cite{CrochemoreLR25addpbs}).

The \emph{Period-Doubling sequence} is obtained with the morphism
$\sigma\colon \tta \to \tta\ttb, \ttb \to \tta\tta$.

The \emph{Chacon sequence} is obtained with the morphism
$\sigma\colon \tta \to \texttt{aaba}, \ttb \to \ttb$.

The \emph{de Bruijn sequence} of order $n$, here on a two-letter alphabet,
is a cyclic sequence in which every possible factor of length $n$
occurs exactly once. Its length is $2^n$. Its non-cyclic version has length $2^n+n-1$, and a smallest attractor has a size close to $2^n/n$.

The \emph{Oldenburger–Kolakoski sequence}, sometimes also known as the  \emph{Kolakoski sequence},
is an infinite sequence of symbols $\{\texttt{1},\texttt{2}\}$ that is the sequence of run lengths in its own run-length encoding.
Few things are known about this sequence. It is not known to be uniformly recurrent or that the frequency of $\texttt{1}, \texttt{2}$ is $1/2$. See \cite{BoissonJametMarcovici2024} for an analysis of some of its statistical properties.

The \emph{Powers of $2$ sequence} is the characteristic sequence of the powers of $2$, that is the infinite sequence 
$x= (x_n)_{n \geq 0}$ on the two-letter alphabet $\{\texttt{0}, \texttt{1}\}$ defined by $x_n  = 1$ if $n + 1$ is a power of $2$ and $0$ otherwise.

\medskip
Table \ref{table.1} gives in each column the size of the attractor computed by  \Algo{GetAnAttractor} on all prefixes of size  $n = 2^i$, for $0 \leq i \leq 2^{21}$, of the infinite sequence corresponding to the column.
\begin{table}[!ht]
\begin{tabular}{|| l | l | l | l | l ||} 
 \hline
\footnotesize{Length $n$} & \footnotesize{Fibonacci} & \footnotesize{Thue-Morse} & \footnotesize{Period-Doubling} & \footnotesize{Chacon}  \\ 
 \hline\hline
1 & 1 &  1&  1  & 1\\ 
 \hline
2 &  2&  2&  2  & 1\\ 
 \hline
 4 
 &  2&  2&  2 & 2 \\ 
 \hline
 8 &  2&  3&  2   &2 \\ 
 \hline
 16 &  2&  4&  2  & 3 \\ 
 \hline
 32 &  2&  5&  2  &  3\\ 
 \hline
64 &  2&  5&  2   & 4\\ 
 \hline
128 &  2&  5&  2   & 5\\ 
 \hline
256 &  2&  5&  2   & 5\\ 
 \hline
512 &  2&  5&  2   & 6\\
 \hline
1024 &  2&  5&  2   & 7\\ 
 \hline
2048 &  2&  5&  2   & 7\\
 \hline
 4096 &  2&  5&  2   & 8\\
 \hline
8192&  2&  5&  2   & 8\\
 \hline
16384 &  2&  5&  2   & 9\\
 \hline
32768 &  2&  5&  2   & 10\\
 \hline
65536 &  2&  5&  2   & 10\\
 \hline
131072 &  2&  5&  2   & 11\\
 \hline
262144 &  2&  5&  2   & 12\\
 \hline
524288 &  2&  5&  2   & 12\\
 \hline
1048576 &  2&  5&  2   & 13\\
 \hline
2097152 &  2&  5&  2   & 13\\ 
 \hline
\end{tabular}
\caption{\Algo{GetAnAttractor}:  size of the attractor obtained experimentally for linearly recurrent substitutive sequences.}
\label{table.1}
\end{table}

The results presented in Table \ref{table.1} tend to show the existence of attractors of size $O(1)$.
This aligns with the bound established in \cite{SchaefferShallit2024}, as these sequences are linearly recurrent. The minimum attractor size is achieved for the Fibonacci and Period-Doubling sequences.

Table \ref{table.2} displays the size of the attractor, computed by  \Algo{GetAnAttractor}, for the de Bruijn word of order $j$, $0 \leq j \leq K=21$. The de Bruijn word of order $j$ is drawn on a two-letter alphabet and contains a unique occurrence of all factors of length $j$. The other columns give 
the size of the computed attractor for each prefix of length 
$n = 2^i$, for $0 \leq i \leq 2^{K}$, of the infinite sequence corresponding to the column.
The random sequence of size $2^{K}$ is a pseudorandom sequence on two letters generated by the Java \texttt{Math.random} method.
\begin{table}[!ht]
\begin{tabular}{|| l | l | l | l | l | l | l||} 
 \hline
\footnotesize{Length $n$} & \footnotesize{de Bruijn} & \footnotesize{Kolakoski} & \footnotesize{Pow. of $2$} & \footnotesize{Rand.} & \footnotesize{$\lfloor \log(n) \rfloor$} & \footnotesize{$\lfloor n/\log(n) \rfloor$} \\ 
 \hline\hline
1 &  1&1&1  &  1&  0& - \\ 
 \hline
2 &  2&2& 1 &  2&  0&  2   \\ 
 \hline
 4 &  2&2& 2 &  2&  1&  2    \\ 
 \hline
 8 &  3&3&  3&  2&  2&   3  \\ 
 \hline
 16 &  4&2& 4 &  3&  2&   5   \\ 
 \hline
 32 &  7&4&  5&  5&  3&  9  \\ 
 \hline
 64 &  11&5&  6&  6&  4&    15  \\ 
 \hline
 128 & 19& 7&  7&  8&  4&     26 \\ 
 \hline
  256 &  33&12& 8 & 14 &  5&     46 \\ 
 \hline
  512 &  58&19&  9&  28&  6&     82\\ 
 \hline
  1024 &  103&32&  10& 50 &  6&    147  \\ 
 \hline
  2048  &  187&40&  11&  92&  7&   268   \\ 
 \hline
4096  &  343&56   &  12&  165& 8&    492  \\ 
 \hline
 8192 &  631&88  &  13&  319& 9&    909  \\ 
 \hline
 16384 &  1173&132&  14 &  579&9&  1688    \\ 
 \hline
 32768 &  2186&206&   15&  1080&10&  3151    \\ 
 \hline
65536 &  4101&330&    16&  2017&11&   5909   \\ 
 \hline
131072 &  7711&518&    17&  3819&11&    11123   \\ 
 \hline
262144 &  14572&785&  18&  7185&12&    21010 \\ 
 \hline
524288 &  27595&1229&   19&  13653&13&    39809 \\ 
 \hline
1048576 &  52445&1878&    20&  25948&13&    75638 \\ 
 \hline
2097152 &  99868 &2930&    21&  94305&14&    144073  \\ 
 \hline
\end{tabular}
\caption{\Algo{GetAnAttractor}:  size of the attractor obtained experimentally for some non-linearly recurrent sequences.}
\label{table.2}
\end{table}

\subsubsection*{Comparison with bounds on the size of a smallest attractor}

The size of a smallest attractor for large families of sequences
is studied in \cite{SchaefferShallit2024}.
Let $x$ be an infinite sequence. Let $\gamma_x(n)$ denote the size of a smallest attractor
of the prefix of length $n$ of $x$.

It is shown in \cite{SchaefferShallit2024} that $\gamma_x(n) = O(1)$ for a linearly recurrent sequence $x$. This contrasts with $\gamma_x(n) = O(\log(n))$ for an automatic sequence $x$.
Further, $\gamma_x(n) = \Theta(1)$ or $\gamma_x(n) = \Theta(\log(n))$, and
it is decidable whether $\gamma_x(n)$ is $\Theta(1)$ or $\Theta(\log(n))$ 
using an automaton that generates the automatic sequence.
This result is extended in \cite[Theorem 23]{DBLP:journals/jcta/CassaigneGRRSS24}, where it is proven
that $\gamma_x(n) = O(1)$ or $\gamma_x(n) = O(\log(n))$ for all fixed point $x$ of substitutions, and that it is decidable whether the former or the latter occurs.

In \cite{SchaefferShallit2024}, the authors give a greedy algorithm to find an attractor of size 
$O(A_x \log n)$, where the \emph{appearance constant} $A_x$ is finite.
The appearance constant is the smallest constant $C$ for which the prefix of length $Cn$
of $x$ contains all factors of length $n$ of $x$. It is finite for automatic sequences.
A linearly recurrent sequence is an infinite sequence $x$ for which
there is a constant $R_x$ that satisfies: for each factor $u$ of $x$, there is a word $r$ of length at most $R_x|u|$ for which $ru$ contains two occurrences of $u$.
Note that a linearly recurrent infinite sequence admits a finite appearance constant.
Indeed, it is known that $A_x \leq R_x + 1$.
The following proposition shows that the algorithm \Algo{GetAnAttractor} also computes 
an attractor of size $O(A_x \log n)$. 

\begin{proposition} \label{proposition.size1}
Let $x$ be an infinite sequence with a constant finite appearance $A_x$. Then, Algorithm \Algo{GetAnAttractor} applied to the prefix $w=x_{[1, n]}$ of $x$ computes an attractor of size $O(A_x\log n)$.
\end{proposition}

\begin{proof}
Let $w=x_{[1, n]}$. Let $k$ be the integer such that $2^{k-1} \leq n < 2^{k}$.
We consider intervals $I_s =\llbracket 1, s \rrbracket$ of possible lengths of factors of $w$. 
Note that $\llbracket 1, n \rrbracket \subseteq I_{2^{k}}$.

We consider a version of the algorithm \Algo{GetAnAttractor}, called \Algo{GetAnAttractorP}.
A run of \Call{GetAnAttractorP}{\mathcal{S}(w), s} applies to 
the prefix $w=x_{[1, n]}$. The change is done on lines \ref{algo6-line-7}-\ref{algo6-line-12}, where only the states $q$ such that $\pos[q] = i$ and $\ell[q] \leq s$ are scanned. Hence, the condition of line \ref{algo6-line-9} is changed to "$d[q] \geq \ell[q]$ and $\ell[q] \leq s$". 

Let $P^{(s)} = (k^{(s)}_{1}, k^{(s)}_2, \ldots, k^{(s)}_{m_s})$ be the (decreasing) sequence of positions that is the output of \Call{GetAnAttractorP}{\mathcal{S}(w), s}. 
Hence, for each new position $k$ added in $P^{(s)}$ at line 13 during the execution of \Call{GetAnAttractorP}{\mathcal{S}(w), s}, corresponds at least one state $q$ and a position $i$, strictly less
that the positions already put in $P^{(s)}$, with $\pos[q] = i = k + |u| -1$, $|u| = \ell[q] \leq s$, $x_{[k, i]} = u$.
We denote by $f_s(k)$ the triplet $(q, i, u)$ corresponding to $k$ in $P^{(s)}$.

Let us show that there is at most one position in $P^{(s)}$ strictly between two consecutive positions of $P^{(2s)}$.

Indeed, let $i' < i$ such that $P^{(s)} \cap \, ]k^{(2s)}_{i'+1}, k^{(2s)}_{i'}[ \;= (k^{(s)}_{i}, \ldots, k^{(s)}_{i+m})$ in decreasing order. 
Assume that there are at least two positions in $P^{(s)}$ strictly between $k^{(2s)}_{i'+1}$ and  $k^{(2s)}_{i'}$, that is,  $m \geq 1$. 
Let $(q, r, u) = f_s(k^{(s)}_{i+1})$. Hence, $r < k^{(s)}_{i}$, $\pos[q] = r$,  $|u| = \ell[q] \leq s$, and $x_{[k^{(s)}_{i+1}, r]} = u$. 
By definition of $k^{(2s)}_{i'+1}$ and since $r < k^{(s)}_{i} < k^{(2s)}_{i'}$, this implies $k^{(2s)}_{i'+1} \geq k^{(s)}_{i+1}$, a contradiction.
Moreover, if $k^{(2s)}_{i'} = k^{(s)}_{i-1}$, then $P^{(s)} \cap \; ]k^{(2s)}_{i'+1}, k^{(2s)}_{j}[$ is empty for a similar reason.

Let us now show that the distance between two consecutive positions 
of $P^{(2s)} \cap \; ]k^{(s)}_{i+1}, k^{(s)}_{i}[$ is at least $s+1$.
Assume the contrary. Let $k^{(s)}_{i+1} < k^{(2s)}_{i'+1} <  k^{(2s)}_{i'} < k^{(s)}_{i}$
(see Figure~\ref{figure.size1}).
If $k^{(2s)}_{i'} -  k^{(2s)}_{i'+1} \leq s$, then $f_{2s}(k^{(2s)}_{i'+1}) = (q, r, u)$ with $|u| \leq s$. 
By definition of $k^{(s)}_{i+1}$, this implies that $k^{(s)}_{i+1} \geq k^{(2s)}_{i'+1}$, a contradiction. 

\begin{figure}
\begin{center} %
    \begin{tikzpicture}
\draw (-1,0) -- (10,0);
\draw (-1,1) -- (10,1);
\draw (0,0) node[below right] {} node{$\bullet$};
\draw (2,0) node[below right] {} node{$\bullet$};
\draw (3,0) node[below right, below] {$k^{(2s)}_{i'+1}$} node{$\bullet$};
\draw (4,0) node[below right, below] {$k^{(2s)}_{i'}$} node{$\bullet$};
\draw (6,0) node[below right] {} node{$\bullet$};
\draw (8,0) node[below right] {} node{$\bullet$};
\draw (1,1) node[below right, above] {$k^{(s)}_{i+1}$} node{$\circ$};
\draw (5,1) node[below right, above] {$k^{(s)}_{i}$} node{$\circ$};
\draw (8,1) node[below right] {} node{$\circ$};

    \end{tikzpicture}
\end{center}
\caption{\label{figure.size1} Positions of $P^{(s)}$ and $P^{(2s)}$ in $w$.}
\end{figure}

Since each factor of length at most $2s$ is a factor of the suffix of length 
$2sA_x$ of $w$, all positions $r$ such that $f_{2s}(k) = (q, r, u)$ for some $k$ in $P^{(2s)}$
are inside $\llbracket 1, 2sA_x \rrbracket$ by definition of the table $\pos$.
Moreover, any two consecutive positions of $P^{(2s)}$ within two consecutive positions of $P^{(s)}$
are separated by a distance greater than $s$. 
This implies that the total number of all $k^{(2s)}_{i'}$ is at most $|P^{(s)}| + 2A_x$.
Hence, $|P^{(2s)}| \leq |P^{(s)}| + 2A_x$. 

For $i = 1$,  $I_1 = \llbracket 1 \rrbracket$, $\Algo{GetAnAttractorP}$ outputs $P_1$ of size at most $|A|$.
We get $|P^{(2^k)}| \leq |P^{(1)}| + 2kA_x$. Hence, the size of the attractor $P$ produced by $\Algo{GetAnAttractorP}$
is the size of $P^{(2^k)}$, at most $|P^{(1)}| + 2A_x (\log(n) + 1) = O(A_x\log n)$.

Eventually, the result is also valid for $\Algo{GetAnAttractor}$ that proceeds with a more general condition at line \ref{algo6-line-9}.
\end{proof}

Algorithm $\Algo{GetAnAttractor}$ applied to the prefix $w=x_{[1, n]}$ of a linearly recurrent sequence $x$ may not compute an attractor of size $O(1)$. Indeed, the Chacon sequence is linearly recurrent, but the size of the attractor obtained with $\Algo{GetAnAttractor}$ does not seem to have experimentally a size bounded by a constant (see Table~\ref{table.1}). 
This is due to the fact that the algorithm transmits information only backwards while scanning the positions in the word. The linear recurrence property needs to transmit some information forward to reduce the size of the attractor.

\section{Conclusion}

The DAWG or Suffix automaton data structure reveals to be a powerful tool for addressing questions on attractors.
This is far from surprising because the DAWG gives a rather direct access to all the factors of a word, like suffix trees and suffix arrays.
But the suffix link tree of the DAWG, different from its analogues in the other data structures, clarifies the required properties used to design the appropriate algorithms.

The present work opens the way to the study of attractors of specific types of words and to the design of efficient algorithms to compute their smallest attractors when feasible. 

\subsection*{Acknowledgements}

We warmly thank the reviewers for their comments and suggestions. They help us clarify several central elements of the article.

\bibliography{attractors}

\end{document}